\documentclass[12pt,3p]{article}

\usepackage[english]{babel}
\usepackage{float, color}
\usepackage{graphicx,epstopdf}
\usepackage{amsmath,amssymb,amsfonts,mathrsfs} 
\usepackage{multicol,subfigure} 
\usepackage{comment}
\usepackage{cite}
\usepackage{adjustbox}
\usepackage{pdflscape}
\usepackage{array}
\usepackage{booktabs}
\usepackage{tabularx,ragged2e,booktabs}
\usepackage{amsthm}
\usepackage{hyperref}
\usepackage{setspace}
\usepackage[utf8]{inputenc}
\usepackage{algorithm}
\usepackage{algorithmic}

\newtheorem{theorem}{Theorem}

\usepackage[switch]{lineno}
\usepackage[absolute,showboxes]{textpos}

\begin{document}


\title{Estimating Shell-Index in a Graph with Local Information}

\author{%
  Akrati Saxena*\\akrati.saxena@iitrpr.ac.in
  \and S. R. S. Iyengar*\\sudarshan@iitrpr.ac.in
  \and *Department of Computer Science and Engineering,\\ Indian Institute of Technology Ropar, India
  }
\date{}

\maketitle

\begin{abstract}
For network scientists, it has always been an interesting problem to identify the influential nodes in a given network. The k-shell decomposition method is a widely used method which assigns a shell-index value to each node based on its influential power. The k-shell method requires the global information of the network to compute the shell-index of a node that is infeasible for large-scale real-world dynamic networks. In this work, we propose a method to estimate the shell-index of a node using its local information. We also propose hill-climbing based approach to hit the top-ranked nodes in a small number of steps. We further discuss a method to estimate the rank of a node based on the proposed estimator.

\end{abstract}


\section{Introduction}

In real-world networks, hierarchical organization of networks gives rise to the core-periphery structure. The concept of the core-periphery structure was first proposed by Borgatti and Everett \cite{borgatti2000models} in 2000. The core is a densely connected nucleus of the network that is surrounded by sparsely connected periphery nodes. Core nodes are highly connected with each other and also highly connected with periphery nodes. Authors also proposed a method to detect the core-periphery structure based on the adjacency matrix division, such that the number of edges is maximum among core nodes that is preceded by the number of edges between core and periphery, and periphery nodes have very less number of edges among themselves. The complexity of the proposed matrix division method is very high, and so, it is infeasible for large-size real-world networks. Carmi et al. employed the k-shell decomposition algorithm first time to study the core-periphery structure in the Internet network and classified the network into three categories \cite{carmi2007model}.

The k-shell decomposition method was proposed by Seidman and it is a well-known method in social network analysis to identify the tightly knit group of influential core nodes in unweighted networks \cite{seidman1983network}. This algorithm works by recursively pruning the nodes from lower degree to higher degree. First, we recursively remove all nodes of degree 1, until there is no node of degree 1. All these nodes are assigned shell-index $k_s=1$. In the next step, we remove all nodes of degree 2 if any. While pruning the degree 2 nodes, new nodes having degree 2 or less can be generated and will also be removed in the same iteration. All these nodes will be assigned shell-index $k_s=2$. Similarly, nodes of degree 3, 4, 5, ... are pruned step by step. When we remove nodes of degree $k$, if there appears any node of degree less than $k$, it will also be removed in the same step. All these nodes are assigned shell-index $k$. This method thus divides the entire network into shells and assigns a shell-index to each node. The shell-index increases as we move from the periphery to the core and the higher shell-index represents the higher coreness. The innermost shell has the highest shell-index $k_{max}$ and is called the core of the network. Batagelj and Zaver{\v{s}}nik proposed an order $O(m)$ algorithm to compute the coreness of all nodes, where $m$ is the total number of edges in the network \cite{batagelj2011fast}.


Many studies have shown that core nodes are highly influential nodes in a network.  In the entire discussion, influential nodes refer to the nodes having the higher spreading power. Kitsak et al. \cite{kitsak2010identification} showed that information spreads faster if the spreading is started from a core node rather than a periphery node. Saxena et al. showed the importance of core nodes in information diffusion on networks with meso-scale structures \cite{saxena2015understanding}. The results conclude that the information becomes viral once it hits the core nodes and spreads into multiple communities through the core. Pei and Maske studied the correlation of various centrality measures with the influential power of the nodes on real-world networks and observed that shell-index is an effective centrality measure to identify the most influential nodes \cite{pei2013spreading, pei2014searching}. 

Shell-index is a widely used method to compute the influential power of a node, but it has its disadvantages. Firstly, to compute the shell-index of a node using the k-shell decomposition method, we need the entire network, which is infeasible for large-scale dynamic networks. Secondly, the k-shell method assigns the same index values to many nodes which actually might have different influential power, as shown in \cite{zareie2018hierarchical, wang2016fast, zeng2013ranking}.

In this work, we show that the shell-index value of a node can be estimated using its $h^2-index$ which can be computed using the local neighborhood information of the node. We study the correlation of shell-index and $h^2-index$ with the spreading power of a node and observe that the $h^2-index$ has good correlation with the spreading power and can be used in real-world networks to identify the influential nodes. We further show that the $h^2-index$ has better monotonicity than the shell-index.

In most of the real-world applications, we only need to identify top influential nodes to spread the information without having the global information of the network. We implement hill-climbing based algorithms using the proposed shell-index estimator to identify the top-ranked nodes in a small number of steps. The detailed algorithm and results are discussed in Section~\ref{kshellsec2}. Next, we propose a heuristic method to estimate the influential rank of a node without computing the index value of all the nodes. The proposed method is discussed in Section~\ref{kshellsec3}.

The main contributions of the work are as follows.
\begin{enumerate}
\item Discuss a method to compute the shell-index of a node using the shell-index value of its neighbors.
\item Propose an estimator for the shell-index value of a node.
\item Hill-climbing based methods to identify top-ranked nodes using the proposed estimator.
\item A heuristic method to estimate the percentile rank of a node based on the proposed estimator.
\end{enumerate}

The rest of this paper is organized as follows. In Section~\ref{kshellrelwork}, we discuss the related literature. Section~\ref{prelim} covers the preliminary definitions. In Section~\ref{kshellsection1}, we discuss the estimation of shell-index using $h^2-index$ and experimental results. In Section~\ref{kshellsec2}, we discuss hill-climbing based approaches to identify the top-ranked nodes and their simulation on real-world networks. In Section~\ref{kshellsec3}, we discuss a heuristic method to estimate the percentile rank of a node. Section~\ref{conclusion} concludes the paper with future directions.

\section{Related Work}\label{kshellrelwork}


In recent years, the problem of identifying influential nodes has attracted researchers from different areas like computer science, epidemiology, biology, statistics, etc. Researchers have defined various centrality measures to compute the importance of a node based on its characteristics and the given application context. These centrality measures can be categorized as local centrality measures and global centrality measures. Centrality measures that can be computed using local neighborhood information of the node are called local centrality measures like degree centrality \cite{shaw1954some}, h-index \cite{hirsch2005index}, semi-local centrality \cite{chen2012identifying} etc. The centrality measures that require the entire network for their computation are called global centrality measures like closeness \cite{sabidussi1966centrality}, betweenness \cite{freeman1977set}, eigenvector \cite{stephenson1989rethinking}, shell-index \cite{seidman1983network}, PageRank \cite{brin1998anatomy} etc. The computational complexity of global centrality measures is very high, and it depends on the network size.

The use of a specific centrality measure to identify important nodes is highly application dependent. Recently Kitsak et al. showed that shell-index is highly correlated with the spreading power of a node \cite{kitsak2010identification}. These results have been supported by many other studies that show that the core nodes are highly influential and the information flows into multiple communities through the core \cite{saxena2015understanding, gupta2016modeling}.

Initially, the k-shell decomposition method was defined for undirected unweighted networks, but recently it has been extended to different types of networks. Garas et al. extended the k-shell method to identify the core-periphery structure in weighted networks \cite{garas2012k}. They defined the weighted degree that considers both the degree as well as the weights of the connected edges. Then, the weighted degree is used while applying the k-shell decomposition method. Eidsaa et al. also proposed a method to identify the core-periphery structure in weighted networks where they only consider the strength of the nodes while pruning the nodes in each iteration \cite{eidsaa2013s}. This method is known as s-shell decomposition method or strength decomposition method. Wei et al. proposed an edge-weighting k-shell method where they consider both the degree as well as the edge-weights, and the edge weight is computed by adding the degree of its two endpoints \cite{wei2015weighted}. The importance of both of these parameters can be set using a tuning parameter which varies from 0 to 1.

Researchers have shown that the shell-index can be efficiently used to identify the influential nodes in a network \cite{kitsak2010identification, gupta2016modeling}. However, the k-shell method assigns the same index values to many nodes which actually might have different influential power \cite{zareie2018hierarchical, wang2016fast, zeng2013ranking}. Zeng et al. modified the k-shell decomposition method and proposed a mixed degree decomposition (MDD) method which considers both the residual degree and the exhausted degree of the nodes while assigning them index values \cite{zeng2013ranking}. Liu et al. proposed an improved ranking method that considers both the shell-index of the node and its distance with the highest shell-index nodes \cite{liu2013ranking}. The proposed method computes the shortest distance of all nodes with the highest shell-index nodes, so, it has high computational complexity. Liu et al. showed that in real-world networks some core-like groups are formed which are not true-core \cite{liu2015improving}. The nodes in these groups are tightly connected with each other but have very few links outside. Based on this observation, the authors filtered out redundant links which have low diffusion power but supports non-pure core groups to be formed, and then apply k-shell decomposition method. Authors show that the coreness computed on this new graph is a better measure of influential power and it is highly correlated with the spreading power computed using SIR model in the original graph.

Researchers have also proposed hybrid centrality measures by combining the shell-index with other centrality measures. Hou et al. introduced the concept of all-around score to find influential nodes \cite{hou2012identifying}. All-around score of a node is defined as, $Score=\sqrt{\left \| d \right \|^2 + \left \| C_B \right \| ^2 +\left \| k_s \right \|^2 }$, where $d$ is the degree, $C_B$ is the betweenness centrality, and $k_s$ is the shell-index of the node. The degree takes care of local connectivity of the node, betweenness takes care of shortest paths that represent global information, and shell-index represents the position of the node with respect to the center. The total time complexity of the complete process is $O(n \cdot m)$, as it depends on the complexity of betweenness centrality that has the highest complexity. 

Basaras et al. proposed a hybrid centrality measure based on degree and shell-index and showed that it works better than the traditional shell-index \cite{basaras2013detecting}. Bae and Kim proposed a method where the centrality value of a node is computed based on the shell-index value of its neighbors; it thus considers both degree and shell-index value of the nodes \cite{bae2014identifying}. The results show that the proposed method outperforms other methods in scale-free networks with community structure. Tatti and Gionis proposed a graph decomposition method that considers both the connectivity as well as the local density while the k-shell decomposition method only considers the connectivity of the nodes \cite{tatti2015density}. The running time of the proposed algorithm is $O(n^2 \cdot m)$. They further proposed a linear-time algorithm that computes a factor-2 approximation of the optimal decomposition value. All the discussed centrality measures have better monotonicity, but they require global information of the network to be computed; and so, they are not favorable in large-scale networks.

Real-world networks are highly dynamic, and it is infeasible to re-compute the shell-index of each node for every single change in the network. Li et al. proposed a method to update the shell-index value of the affected nodes, whenever a new edge is added or deleted from the network \cite{li2014efficient}. Jakma et al. proposed the first continuous, distributed method to compute shell-index in dynamic graphs \cite{jakma2012distributed}. Pechlivanidou et al. proposed a distributed algorithm based on MapReduce to compute the k-shell of the graph \cite{pechlivanidou2014mapreduce}. Montresor et al. proposed an algorithm to apply the k-shell method on live distributed systems \cite{montresor2013distributed}. The execution time of the proposed algorithm is bounded by $1+\sum_{u \in V}[d(u)-k_s(u)]$, and it gives 80 percent reduction in the execution time on the considered real-world networks. Dasari et al. proposed a k-shell decomposition algorithm called ParK that reduces the number of random access calls and the size of the working set while computing the shell-index in larger graphs \cite{dasari2014park}. They further proposed a faster algorithm which involves parallel execution to compute the k-shell in larger graphs. Sariyuce et al. proposed the first incremental k-shell decomposition algorithms for streaming networks \cite{sariyuce2013streaming}. The proposed method has million times speed-up than the original k-shell method on a network having 16 million nodes. Lu et al. explored the relationship between the degree, h-index \cite{hirsch2005index}, and coreness of a node \cite{lu2016h}. They showed that the h-index family of a node converges to the coreness of the node.

K-shell method has been widely used in literature to study different networks. Catini et al. used shell-index to identify clusters in PubMed scientific publications \cite{catini2015identifying}. To identify the clusters, a graph is created where the nodes are the publications, and there is an edge between two nodes if the distance between the locations of the corresponding researchers is less than the threshold. In the experiments, authors have taken the threshold 1 km. Based on the k-shell decomposition, authors categorized the cities into monocore and multicore. Later on, the journal impact factors are used to quantify the quality of research of each core. Results show that k-shell decomposition method can be used to identify the research hub clusters.

Core-periphery structure has been studied in the wide variety of networks, such as financial network \cite{fricke2015core, barucca2016disentangling}, human-brain network \cite{bassett2013task, park2013structural}, nervous system of C. elegans worm \cite{chatterjee2007understanding}, blog networks \cite{obradovic2009journey}, scientific publication network \cite{catini2015identifying}, collaboration network \cite{leydesdorff2013international, hu2008visual}, protein interaction networks \cite{luo2009core}, communication network of software development team \cite{crowston2006core, amrit2010exploring, setia2012peripheral, cataldo2008communication}, hollywood collaboration network \cite{cattani2014insiders}, language network \cite{fedorenko2014reworking}, YouTube social interaction network \cite{paolillo2008structure}, metabolic networks \cite{zhao2007modular}, etc. 
Karwa et al. proposed a method to generate all graphs for a given shell-index sequence \cite{karwa2017statistical}.


In one of our works, we studied the properties of the core like its size and density \cite{saxena2016evolving}. We further studied how the core is evolved with time and how the core nodes acquire a specific position in the network. We observed that making more connections with the existing core nodes help a node to shift into the core regardless of the total number of connections. Based on our observations, we proposed evolving models for both unweighted and weighted networks having community and core-periphery structure. We further studied the core-periphery structure in multilayered terrorist networks \cite{gera2017three, miller2018discovering} and proposed an evolving model to generate synthetic multilayered networks having similar characteristics \cite{adeniji2017generative}.

With time the size of real-world networks is increasing exponentially, so, it is infeasible to collect the entire network to study its global properties. Researchers have focussed to propose fast and efficient methods to identify influential nodes and their ranking in the given network \cite{saxena2017global}. Saxena et al. have proposed methods to estimate the degree rank of a node using power-law degree distribution \cite{saxena2015rank, saxena2015estimating} and sampling techniques \cite{saxena2017observe, saxena2017degree, saxena2018estimating}. In \cite{saxena2017fast, saxena2017afaster}, authors proposed heuristic methods to fast estimate the closeness rank of a node. In this work, we will discuss a heuristic method to estimate the influential rank of a node using local information.


\section{Preliminaries}\label{prelim}


\subsection{H-Index}\label{hindex}

The h-index of a node $u$ ($h-index(u)$) is $h$ if $h$ of its neighbors have degree at least $h$ and there is no subset of $h+1$ neighbors where each node belonging to that subset has the degree at least $h+1$.

The $h^2-index$ of a node $u$ ($h^2-index(u)$) is computed by taking its $h-index$ based on the $h-index$ of its neighbors and not their degrees.

Note: The h-index of a list $l$ ($h-index(l)$) is $h$ if $h$ is the highest number such that $h$ entries of the list are equal to or greater than $h$.

\subsection{SIR Model}\label{secsir}

We use the Susceptible-Infected-Recovered (SIR) spreading model to simulate the spreading process on real-world networks and compute the spreading power of each node. In the SIR model, a node can be in three possible states: $1.$ S (susceptible), $2.$ I (infected), and $3.$ R (recovered). Initially, all nodes are in the susceptible mode except one, which is infected and will start spreading the infection in the network. The infected node will infect each of its susceptible neighbors with infection probability ($\lambda$). If the neighbor gets infected, its status is changed to Infected. Once an infected node contacts all of its neighbors to infect them, its status is changed to Recovered with probability $\mu$. For generality, we set $\mu = 1$. Recovered nodes will neither be infected anymore nor infect others, and they remain Recovered until the spreading stops.

Initially, we infect a single node, and all other nodes are susceptible. Then the infection starts spreading from the seed node to the others through links. The spreading process stops when there is no infected node in the network. The number of recovered nodes is considered the spreading power or spreading capability of the original node. We execute the SIR model 100 times from each node and take the average of the spreading power to compute the final spreading power of the node. The infection probability is taken as $\lambda > \lambda_c$, where $\lambda_c=\left \langle d \right \rangle/(\left \langle d^2 \right \rangle - \left \langle d \right \rangle)$ is the epidemic threshold computed using mean-field theory, where $d$ represents the degree of the node \cite{castellano2010thresholds}.

\section{Section 1: Shell-Index Estimation}\label{kshellsection1}

In this section, we discuss a method to estimate the shell-index using local neighborhood information.

\begin{theorem}\label{kshelltheorem1}
The shell-index of a node $u$ can be computed as $k_s(u)=h-index(k_s(v) | \forall v \in ngh(u))$, where $ngh(u)$ is the set of the neighbors of node $u$.
\end{theorem}

\begin{proof}
Let's assume that h-index of $(k_s(v) | \forall v \in ngh(u))$ is $h$ then there are at least $h$ nodes having shell-index equal to or greater than $h$ as per the definition of h-index.

Now, we will see how the shell-index of a node $u$ will be decided in the k-shell decomposition method. In the k-shell decomposition method, in $i_{th}$ iteration all those nodes are removed who have exactly $i$ connections with the nodes having the shell-index $i$ or greater than $i$. Thus, the node $u$ will be removed in $h_{th}$ iteration as $h$ of its neighbors have shell-index $h$ or greater than $h$. This is nothing but the $h-index$ of node $u$ based on the shell-indices of its neighbors as defined above.
\end{proof}

Next, we explain Theorem~\ref{kshelltheorem1} using examples. The first example is shown in Figure~\ref{exa1}(a) where node $u$ has eight neighbors having shell-indices 1, 2, 3, 3, 4, 6, 8, and 10. Now, we will see how the shell-index of node $u$ will be determined during the k-shell decomposition method. All the iterations are shown in Figure~\ref{exa1}. In the $1_{st}$ iteration, first of its neighbors will be removed and node $u$ will be left with seven connections with the nodes having the shell-indices greater than $1$, so, the node $u$ will not be removed in the first iteration. In the $2_{nd}$ iteration, its second neighbor will be removed as it has shell-index 2, but still, the node $u$ has six connections with the higher shells, so, it will not be removed. In the third iteration, its third and fourth neighbors will be removed as both of these neighbors have shell-index 3. The node $u$ still has four connections with the higher shells, so it will not be removed. In the fourth iteration, 5th of its neighbors having shell-index 4 will be removed, and now the node $u$ has only three connections with the higher shells, so, as per the k-shell decomposition method, node $u$ will also be removed in the fourth iteration. So, the shell-index of node $u$ is 4 that is nothing but the $h-index$ of the shell-indices of its neighbors.



\begin{figure*}[]
     \begin{center}
        \subfigure[Initial subgraph having all connections of node $u$]{%
            \label{fig:first}
            \includegraphics[width=0.45\textwidth]{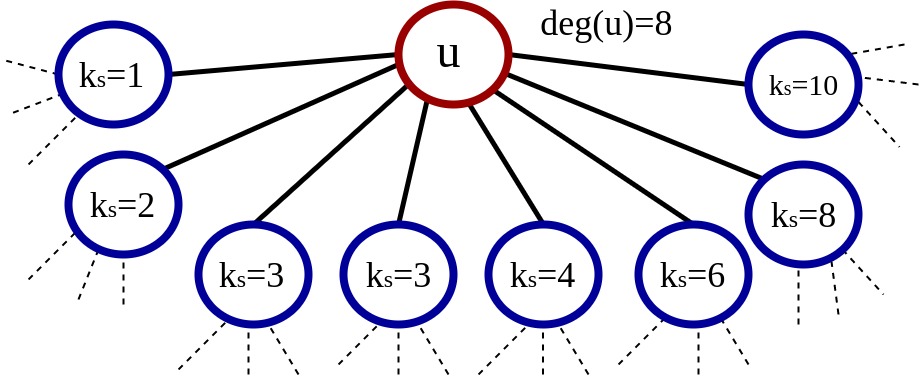}  
        }\quad%
        \subfigure[After Iteration 1]{%
           \label{fig:second}
           \includegraphics[width=0.45\textwidth]{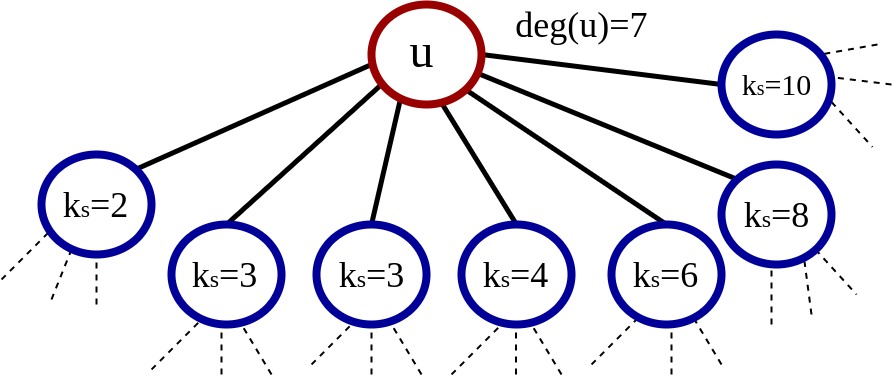}
        }\quad
        \subfigure[After Iteration 2]{%
            \label{fig:first}
            \includegraphics[width=0.45\textwidth]{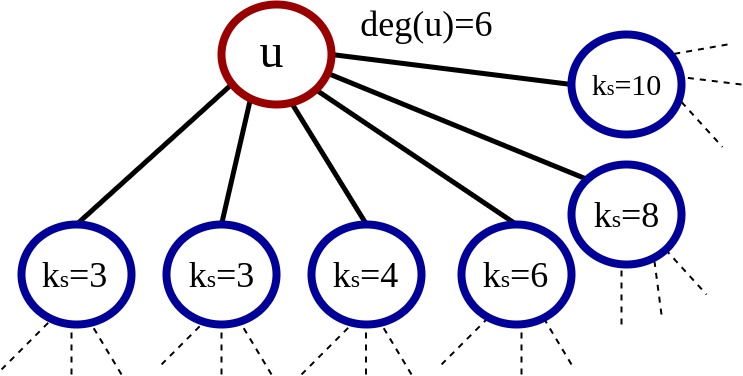}  
        }
        \subfigure[After Iteration 3]{%
            \label{fig:first}
            \includegraphics[width=0.30\textwidth]{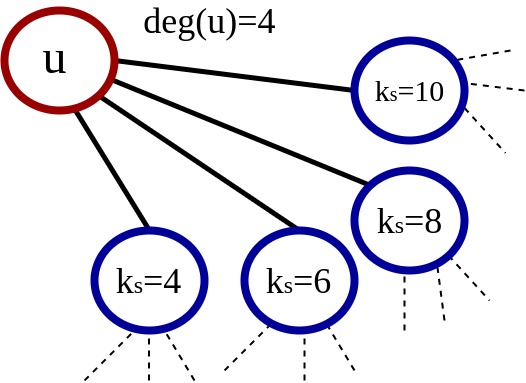}  
        }\quad%
    \caption{Example 1: Estimate shell-index of node $u$ while applying k-shell decomposition algorithm}
   \label{exa1}
   \end{center}
\end{figure*}

Similarly, in Figure~\ref{exa2}, node $u$ has degree 7, and the shell-indices of its neighbors are 1, 2, 2, 3, 3, 3, and 5. Now in the 1st iteration of the k-shell decomposition method, 1st of its neighbors will be removed. Node $u$ still has six connections with the higher shells, so, it will not be removed. In the 2nd iteration, second and third of its neighbors will be removed and node $u$ still has four connections with the higher shells, so, it will not be removed. In the 3rd iteration, 4th, 5th, and 6th of its neighbors will be removed. Node $u$ now has only one connection with the higher shells, so, it will also be removed in the same iteration. Thus, it has shell-index 3.

\begin{figure}[htp]
  \centering
    \includegraphics[width=.55\linewidth]{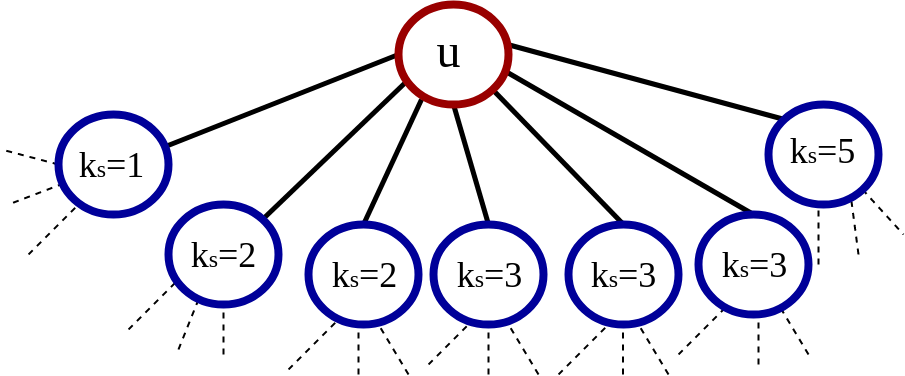}
  \caption{Example 2: A subgraph of a network to explain shell-index computation using shell-indices of neighbors}
  \label{exa2}
\end{figure}

In Figure~\ref{exa3}, node $u$ has degree 6 and the shell-indexes of its neighbors are 10, 11, 11, 13, 15, and 25. In 1st, 2nd, ..., 5th iteration, the node $u$ will not be removed as it has six connections with the higher shells. In the 6th iteration, node $u$ will be removed, so, the shell-index of node $u$ is 6, i.e., the h-index of the shell-indices of its neighbors.

\begin{figure}[htp]
  \centering
    \includegraphics[width=.55\linewidth]{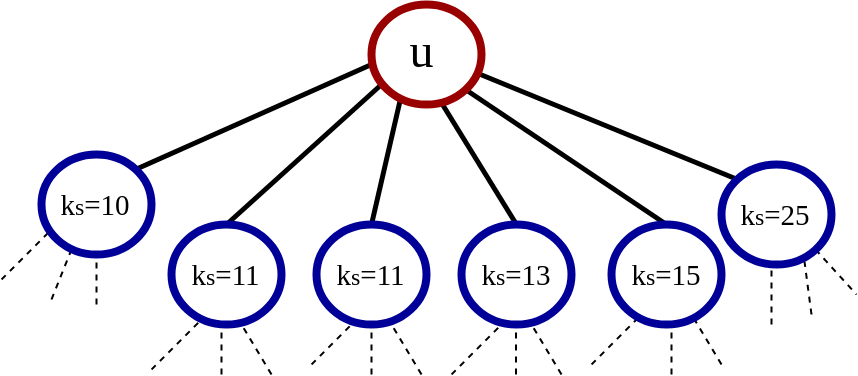}
  \caption{Example 3: A subgraph of a network to explain shell-index computation using shell-indices of neighbors}
  \label{exa3}
\end{figure}

\subsection*{Shell-Index Estimation}

Using Theorem~\ref{kshelltheorem1}, the shell-index of a node can be estimated if the shell-indices of its neighbors are known. However, in real life applications, the shell-indices of the neighbors will not be known. We know that the shell-index of a node is bounded above by its degree, $k_s(v) \leq d(v)$. So, to estimate the shell-index of node $u$, we can consider the degree of its neighbors in place of their shell-index. This is nothing but the $h-index$ of node $u$ as defined in Section~\ref{hindex}. To further improve the estimation, we consider the $h-index$ of its neighbors as $k_s(v) \leq h-index(v) \leq d(v)$ and this is nothing but the $h^2-index$ of the node. Thus the shell-index of a node can be estimated using its $h^2-index$. Next, we will study the performance of the proposed estimator in real-world networks.

Pseudo-code for the proposed method is given in Algorithm~\ref{kshellalgo1}, where $ngh(u)$ is the list of the neighbors of node $u$. The proposed estimator can be further improved if we compute the $h^3-index$ of the node, however, in the Results section, we show that $h^2-index$ itself is a good estimator. It can be computed faster and requires less neighborhood information than to compute the $h^3-index$ of the node.

\begin{algorithm}
\caption{$Compute\_h^2-index(G,u)$}\label{kshellalgo1}
\begin{algorithmic}
\STATE Take a list $ngh\_hindex$, $ngh\_hindex=[\; ]$
\FOR{each $v$ in $ngh(u)$}
\STATE add $h-index(v)$ in $ngh\_hindex$
\ENDFOR
\STATE Return $h-index(ngh\_hindex)$
\end{algorithmic}
\end{algorithm}


\subsection{Results and Discussion}

We study the performance of $h^2-index$ versus shell-index on the following real-world networks.

\subsubsection*{Datasets}

\begin{enumerate}
\item \textbf{Astro-Ph:} This is the co-authorship network of ArXiv's Astrophysics (Astro-ph) publications where authors are represented by nodes, and there is an edge between two nodes if the corresponding authors have published together \cite{leskovec2007graph}. It contains 14845 nodes and 119652 edges.

\item \textbf{Buzznet:} Buzznet is a subgraph extracted from a social networking site that is used to share the photo, journal, and video \cite{zafarani2009social}. It consists of 101163 nodes and 2763066 edges.

\item \textbf{Cond-Mat:} This is the co-authorship network of ArXiv's condensed matter section \cite{leskovec2007graph}. This dataset covers all papers from January 1993 to April 2003 (124 months). It contains 13861 nodes and 44619 edges.

\item \textbf{DBLP:} This is a coauthorship network extracted from DBLP computer science bibliography, where an edge denotes that the authors have common publications \cite{yang2015defining}. This network contains 317080 nodes and 1049866 edges.

\item \textbf{Digg:} This friendship network was extracted from Digg website in 2009 \cite{hogg2012social}. It contains 261489 nodes and 1536577 edges.

\item \textbf{Enron:} This network is the email communication network of the employees of Enron organization from 1999 to 2003 \cite{klimt2004enron}. Nodes of the network are email addresses, and there is an edge between two nodes if they have communicated at least once. The dataset has 84384 nodes and 295889 edges.

\item \textbf{Facebook:} This dataset is an induced subgraph of Facebook \cite{viswanath2009evolution}, where users are represented by nodes and friendships are represented by edges. It contains 63392 nodes and 816831 edges.

\item \textbf{FB-Wall:} This is a network where nodes represent Facebook users, and there is an edge between two users if any one of them posts on the Facebook-wall of another user \cite{viswanath2009evolution}. It contains 43953 nodes and 182384 edges.

\item \textbf{Foursquare:} Foursquare is a location-based social networking software for mobile devices that can be accessed using GPS. This dataset is an induced subgraph of friendships of Foursquare  \cite{zafarani2009social}. It contains 639014 nodes and 3214985 edges.

\item \textbf{Gowalla:} This friendship network is extracted from a location-based social network called, Gowalla \cite{cho2011friendship}. Gowalla was used to share the locations among its users. It contains 196591 nodes and 950327 edges.

\end{enumerate}

The experimental results are shown in Table~\ref{sec1results}. First, we compute the monotonicity of shell-index and $h^2-index$. Ideally, if a node has influential power different from other nodes, it should be assigned a different index value. In the k-shell decomposition method, all the nodes which are pruned at one level are assigned the same shell-index value. Researchers have shown that they can have different influential power and should have been assigned different values \cite{zareie2018hierarchical, wang2016fast, zeng2013ranking}. A better centrality measure should assign the same index value to fewer nodes and allocate more unique values. This characteristic of the measure can be captured using the monotonicity \cite{bae2014identifying}. This is defined as,

\begin{center}
$M(R)=\left( 1-\frac{\sum_{r \in R}n_r(n_r-1)}{n(n-1)} \right) ^2$
\end{center}
where $R$ denotes the ranking values of all the nodes based on any given centrality measure, $n$ represents the size of $R$, i.e., the number of nodes in our case, and $n_r$ represents the number of nodes having rank $r$. If all nodes have the same rank, the monotonicity $(M)$ of the ranking is 0, and it is not a valid ranking measure. If each node has a unique rank then the monotonicity $(M)$ is 1. Results in Table~\ref{sec1results} show that the monotonicity of $h^2-index$ is either the same or slightly better than the shell-index.

Next, we study the correlation of shell-index and $h^2-index$ with the spreading power of the node. The spreading power of a node is computed by executing the SIR model (SIR model is defined in Section~\ref{secsir}) 100 times and taking the average of its spreading powers. In experiments, the infection probability $\lambda$ is taken as $\lambda=\left \langle d \right \rangle/(\left \langle d^2 \right \rangle - \left \langle d \right \rangle) + 0.01$. To study the correlation, we compute Kendall's Tau $(\tau)$, Pearson $(r)$, and Spearman $(\rho)$ correlation coefficients. Results in Table~\ref{sec1results} show that the correlation of $h^2-index$ with the spreading power is either as good as the correlation of shell-index with the spreading power or better.

We further study how the correlation of shell-index and $h^2-index$ with the spreading power changes as we vary the infection probability. Results are shown in Figure~\ref{fig2} for Astro-physics collaboration network and FB-wall social interaction network. The epidemic threshold value ($\lambda_c$) for Astro-Ph and FB-Wall network is 0.02 and 0.04 respectively, so, the considered range of the infection probability is taken $0.05-0.14$, i.e., greater than the $\lambda_c$ for both the networks. The results show that $h^2-index$ has a good correlation with varying infection probabilities. 

We thus conclude that the $h^2-index$ is a better centrality measure and a good estimator of the shell-index in real-world networks. $h^2-index$ has an advantage over shell-index that it is a local centrality measure and can be computed for a node using local neighborhood information without gathering the entire network.

\begin{landscape}
\vspace*{20mm}
\begin{table*}[htp]
\centering
\caption{Performance of Shell-Index ($k_s$) and $h^2-index$ using monotonicity and SIR spreading model}
\label{sec1results}
\resizebox{\columnwidth}{!}{%
\begin{tabular}{|l|l|l|l|l|c|c|c|c|c|c|c|}
\hline
Network & Ref &  Nodes & Edges & \multicolumn{2}{|c|}{Monotonicity} & \multicolumn{3}{|c|}{$k_s$ vs. SIR} & \multicolumn{3}{|c|}{$h^2-index$ vs. SIR} \\ \hline

	&	&	&	&  $k_s$ & $h^2-Index$ & Kendall &  Pearson & Spearman & Kendall &  Pearson & Spearman \\	\hline
																			
Astro-Ph	&	\cite{leskovec2007graph}	&	14845	&	119652	&	0.89	&	0.89	&	0.51	&	0.67	&	0.67	&	0.52	&	0.67	&	0.68	\\ \hline
Buzznet	&	\cite{zafarani2009social}	&	101163	&	2763066	&	0.93	&	0.93	&	0.21	&	0.28	&	0.30	&	0.21	&	0.28	&	0.30	\\ \hline
Cond-Mat	&	\cite{leskovec2007graph}	&	13861	&	44619	&	0.75	&	0.76	&	0.55	&	0.69	&	0.69	&	0.56	&	0.70	&	0.70	\\ \hline
DBLP	&	\cite{yang2015defining}	&	317080	&	1049866	&	0.74	&	0.75	&	0.49	&	0.61	&	0.61	&	0.49	&	0.62	&	0.62	\\ \hline
Digg	&	\cite{hogg2012social}	&	261489	&	1536577	&	0.45	&	0.45	&	0.48	&	0.60	&	0.59	&	0.48	&	0.60	&	0.59	\\ \hline
Enron	&	\cite{klimt2004enron}	&	84384	&	295889	&	0.30	&	0.30	&	0.45	&	0.58	&	0.55	&	0.45	&	0.58	&	0.55	\\ \hline
Facebook	&	\cite{viswanath2009evolution}	&	63392	&	816831	&	0.91	&	0.91	&	0.49	&	0.64	&	0.65	&	0.49	&	0.64	&	0.65	\\ \hline
Fb-Wall	&	\cite{viswanath2009evolution}	&	43953	&	182384	&	0.76	&	0.77	&	0.62	&	0.75	&	0.76	&	0.62	&	0.75	&	0.76	\\ \hline
Foursquare	&	\cite{zafarani2009social}	&	639014	&	3214985	&	0.50	&	0.50	&	0.54	&	0.66	&	0.65	&	0.54	&	0.66	&	0.65	\\ \hline
Gowalla	&	\cite{cho2011friendship}	&	196591	&	950327	&	0.73	&	0.74	&	0.59	&	0.71	&	0.72	&	0.59	&	0.71	&	0.72	\\ \hline
																	
\end{tabular}	
}																	
\end{table*}
\end{landscape}

\begin{figure}[htp]
     \begin{center}
        \subfigure[Astro-Ph Collaboration Network]{%
            \label{fig:first50}
            \includegraphics[width=0.95\textwidth]{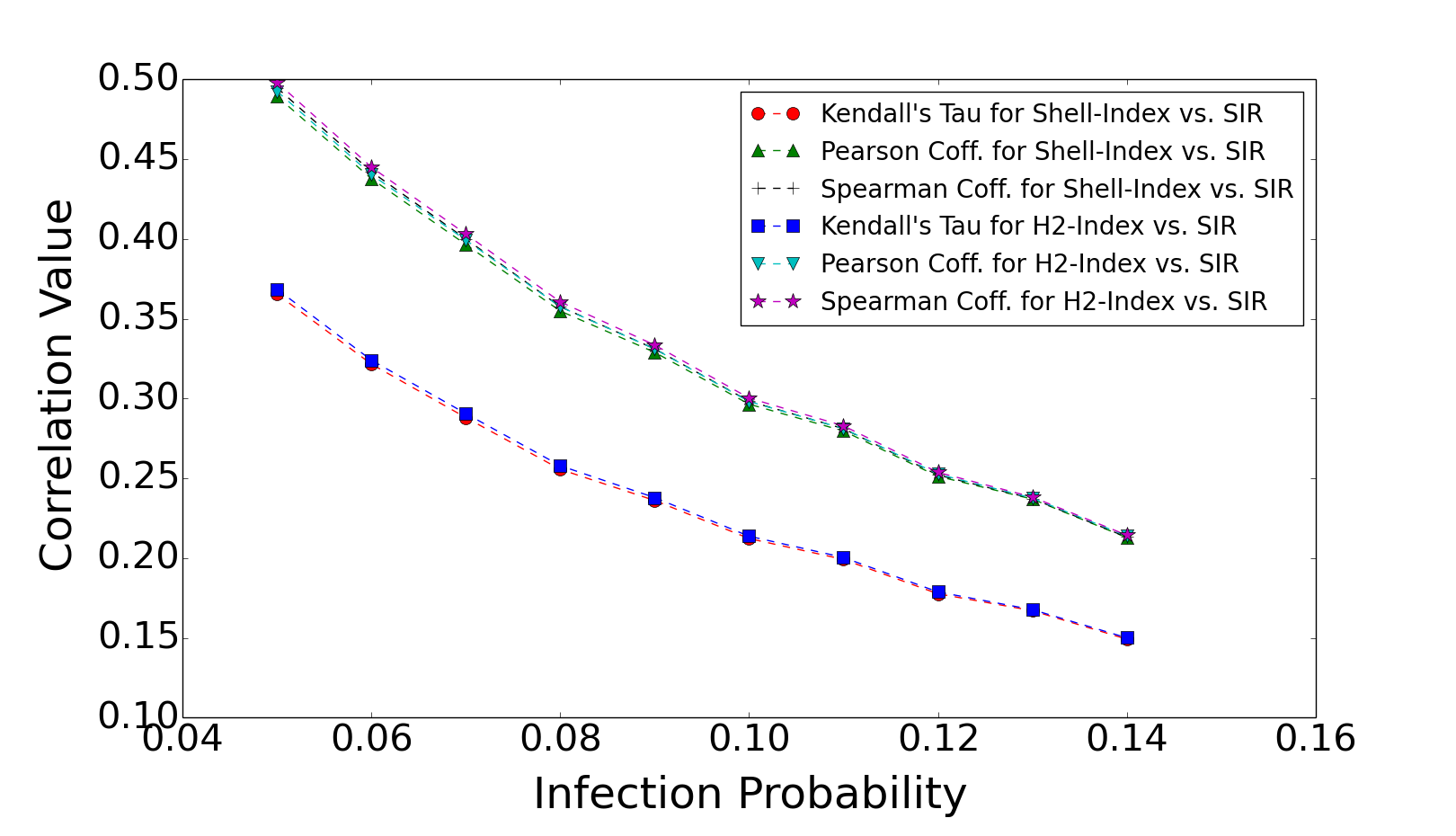}   
        }\quad%
        \subfigure[FB-Wall Social Interaction Network]{%
           \label{fig:second50}
           \includegraphics[width=0.95\textwidth]{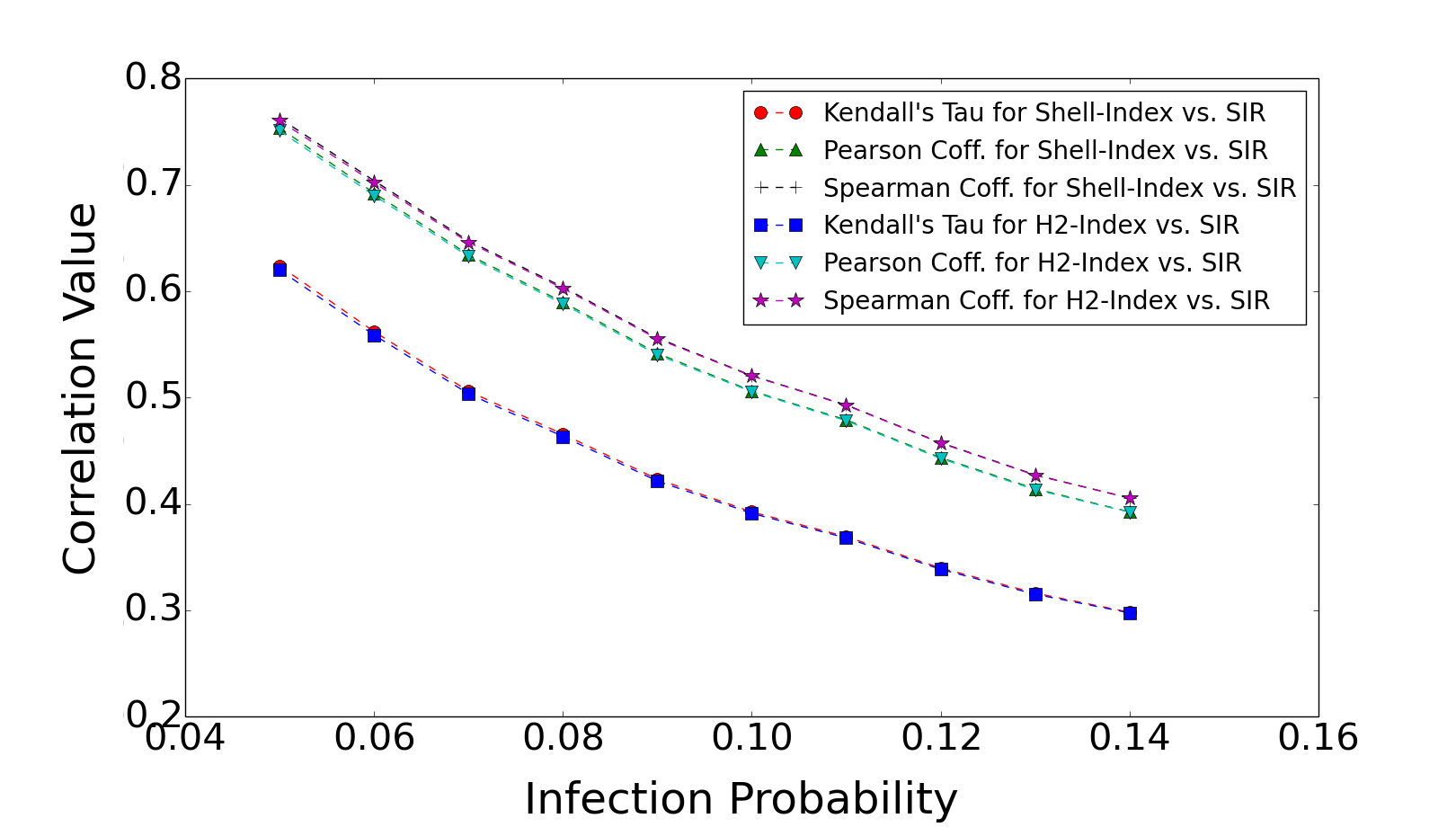}
        }
    \caption{The correlation of shell-Index and $h^2-index$ with the spreading power for varying infection probability on a. Astro-Ph and b. FB-Wall network}
   \label{fig2}
   \end{center}
\end{figure}

\section{Section 2: Hill-Climbing based approach to Identify Top-Ranked Nodes}\label{kshellsec2}

In many real-life applications, we need to identify top influential nodes to spread the information faster. For example, if a marketing company wants to provide free samples of their products, they would like to find out influential people in the network who can help them in popularizing their product. Similarly, if someone wants to spread a virus using the Internet network, she would like to find out and infect a node having the highest spreading power.

In large-scale networks, it is infeasible to collect the entire network to identify the top influential nodes. This motivates us to propose local information based methods for identifying top nodes without having the global information like network size, average degree, the highest $h^2-index$ of the network, and so on. Researchers have proposed various sampling-based methods to estimate global properties of the network like network size \cite{kurant2012graph}, global clustering coefficient \cite{hardiman2013estimating}, average degree \cite{dasgupta2014estimating}, degree distribution of the network \cite{ribeiro2012estimation}, and so on.

In this work, we have shown that the influential power of a node can be computed using its $h^2-index$, i.e. a local measure. For a given node, even if we compute its $h^2-index$ locally, we do not know how influential this node is in the entire network, and whether this node belongs to the highest influential nodes or not. In \cite{gupta2016pseudo}, Gupta et al. proposed Hill-Climbing based methods that can be used to hit the highest shell-index nodes faster in a network. In the proposed method, a random walker starts from a periphery node and move to one of its neighbors having the highest shell-index until a top node is found. The proposed method cannot be applied in practice as the shell-index is a global centrality measure, but the results showed that the proposed hill climbing walk hit the top nodes in a very small number of steps. 

In this work, we have shown that the influential power of a node can be computed using its $h^2-index$, i.e., a local measure. We modify the proposed methods to identify the node having the highest $h^2-index$ in the network. In rest of the discussion, the \textit{top-ranked nodes} refer to the nodes having the highest $h^2-index$ in the network. The algorithm is called $IndexBasedHillClimbing(G,u,k,maxindex)$, and its pseudo code is given in Algorithm~\ref{algo2}. The inputs of the algorithm are $G$, $u$, $k$, and $maxindex$ where $G$ is the given network, $u$ is the seed node from which the crawler starts crawling the network, $k$ is the repeat-count that shows how many times the crawler will restart the walk from a randomly chosen neighbor of the current node if it is stuck to a local maxima, and $maxindex$ is the maximum $h^2-index$ in the given network. A node is called local maxima if its $h^2-index$ is higher than all of its neighbors. The nodes having the highest $h^2-index$ in the network are called global maxima. For the clarification, a local maxima can also be the global maxima.

The algorithm works in the following manner. The crawler starts from the given node $u$, and it moves to one of its neighbors that has not been visited before and has the highest $h^2-index$. The crawler keeps moving until it hits the local maxima. If this node has the highest $h^2-index$, the algorithm exits, else, the crawler jumps to one of its non-visited neighbors uniformly at random, and the repeat-count is increased by one. The same procedure is repeated until the highest $h^2-index$ node is identified or the repeat-count reaches its maximum value. If the algorithm reaches to maximum repeat-count without finding out the highest $h^2-index$ node, it returns ``The algorithm is failed to find out the top-ranked node."  In Algorithm~\ref{algo2}, $ngh(u)$ represents a set of all neighbors of node $u$. $randomchoice(set)$ function returns a random element from the given set.

\begin{algorithm}
\caption{$IndexBasedHillClimbing(G,u,k,maxindex)$}
\label{algo2}
\begin{algorithmic}
\STATE Take a list $visited\_nodes$ and $visited\_nodes = [\;]$
\STATE $num\_of\_steps=0$
\STATE $repeat\_count=0$
\STATE $current\_node=u$
\STATE $next\_node=u$
\STATE add $u$ in $visited\_nodes$
\STATE $flag=True$
\WHILE{$flag == True$}
\FOR {$each \; v \; in \; ngh(current\_node)$} 
\IF {$h^2-index(v) \geq h^2-index(next\_node) \; and \; v \notin visited\_nodes$}
\STATE $next\_node=v$
\ENDIF
\ENDFOR
\IF {$next\_node == current\_node$}
\IF {$h^2-index(next\_node) == maxindex$}
\STATE $flag=False$
\ELSIF {$repeat\_count<k$}
\STATE $next\_node=randomchoice(\{ v \; | \; v \in ngh(current\_node) \; and \; v \notin visited\_nodes \})
$
\STATE $repeat\_count=repeat\_count+1$
\ELSE
\STATE Print ``The algorithm is failed to find out the top-ranked node."
\STATE $flag=False$
\ENDIF
\ENDIF
\IF {$flag==True$}
\STATE add $next\_node$ in $visited\_nodes$	
\STATE $current\_node=next\_node$
\STATE $num\_of\_steps=num\_of\_steps+1$
\ENDIF
\ENDWHILE
\STATE Return $h^2-index(current\_node)$
\end{algorithmic}
\end{algorithm}

The proposed method is further modified as shown in Algorithm~\ref{algo3} named $IndexAndDegreeBasedHillClimbing(G,u,k,maxindex)$. In this method, the crawler will move to one of the highest degree nodes among the non-visited neighbors having the highest $h^2-index$. Beside this, there is one more change; once the algorithm stuck to local maxima, the crawler moves to one of its non-visited neighbors having the highest degree. Intuitively it seems that the highest degree node will have a high probability to be connected with the top-ranked nodes, and so, this algorithm will work faster and better than the first one.

\begin{algorithm}
\caption{$IndexAndDegreeBasedHillClimbing(G,u,k,maxindex)$}
\label{algo3}
\begin{algorithmic}
\STATE Take a list $visited\_nodes$ and $visited\_nodes = [\;]$
\STATE $num\_of\_steps=0$
\STATE $repeat\_count=0$
\STATE $current\_node=u$
\STATE $next\_node=u$
\STATE add $u$ in $visited\_nodes$
\STATE $flag=True$
\WHILE{$flag == True$}
\FOR {$each \; v \; in \; ngh(current\_node)$} 
\IF {$h^2-index(v) \geq h^2-index(next\_node) \; and \; v \notin visited\_nodes$}
\STATE $next\_node=v$
\ENDIF
\ENDFOR\
\FOR {$each \; v \; in \; ngh(current\_node)$} 
\IF {$h^2-index(v) == h^2-index(next\_node) \; and \; deg(v) \geq deg(next\_node) \; and \; v \notin visited\_nodes$}
\STATE $next\_node=v$
\ENDIF
\ENDFOR\
\IF {$next\_node == current\_node$}
\IF {$h^2-index(next\_node) == maxindex$}
\STATE $flag=False$
\ELSIF {$repeat\_count<k$}
\FOR {$each \; v \; in \; ngh(current\_node)$} 
\IF {$deg(v) \geq deg(next\_node) \; and \; v \notin visited\_nodes$}
\STATE $next\_node=v$
\ENDIF
\ENDFOR\
\STATE $repeat\_count=repeat\_count+1$
\ELSE
\STATE Print ``The algorithm is failed to find out the top-ranked node."
\STATE $flag=False$
\ENDIF
\ENDIF
\IF {$flag==True$}
\STATE add $next\_node$ in $visited\_nodes$	
\STATE $current\_node=next\_node$
\STATE $num\_of\_steps=num\_of\_steps+1$
\ENDIF
\ENDWHILE
\STATE Return $h^2-index(current\_node)$
\end{algorithmic}
\end{algorithm}

\subsection*{Discussion}

We implement the proposed methods on real-world networks, and the results are shown in Table~\ref{sec2results}. In the experiments, the algorithm is executed from all non top-ranked nodes (the nodes not having the highest $h^2-index$), and the number of steps taken to hit a top-ranked node are averaged to compute the average number of steps. In experiments, the value of repeat-count $(k)$ is set to 50. In Table~\ref{sec2results}, the average and the standard deviation of the number of steps are shown in $Avg(steps)$ and $Std(steps)$ columns. The $Avg(count)$ shows the average of the number of repeat-count that algorithm takes when it is stuck to a local maxima. The $Algo\;Failed(\%)$ shows how many times the algorithm has not succeeded to hit the top-ranked node in the given repeat-count. The algorithm might succeed if we increase the value of repeat-count.

The results show that on an average, a top-ranked node can be reached in very few steps (3-213 in the considered datasets), and the average value of the repeat-count is 0-10. The algorithm is failed in very few cases for the repeat-count 50. In the case of $IndexAndDegreeBasedHillClimbing(G,u,k,maxindex)$ algorithm, the average number of steps and the average repeat-count is reduced but the probability of failing the algorithm is increased in the given repeat-count as the crawler always moves to higher degree nodes and ends up following the same path that leads to the local maxima. In DBLP network, the probability of failure is much higher when we apply degree based hill-climbing approach as the algorithm is mostly stuck to a local maxima due to following the same route and not able to hit the global maxima in the given repeat-count. This highly depends on the network structure and not on its size.

When we apply these algorithms in practice, we do not know the highest $h^2-index$ ($maxindex$), the algorithms, therefore, will have to be repeated few more times to make sure that the node returned by the algorithm is the actual global maxima and not the local maxima. The results show that a smaller value of $repeat-count$ will suffice the purpose. To increase the efficiency, the crawlers can be started from a few randomly chosen nodes to avoid the local maxima and hit the top-ranked nodes with a high probability.

\begin{table*}[htp]
\centering
\caption{Results for $IndexBasedHillClimbing$ and $IndexAndDegreeBasedHillClimbing$ algorithms}
\label{sec2results}
\resizebox{\columnwidth}{!}{%
\begin{tabular}{|l|c|c|c|c|c|c|c|c|c|}
\hline
Network &  Nodes & \multicolumn{4}{|c|}{$IndexBasedHillClimbing$} & \multicolumn{4}{|c|}{$IndexAndDegreeBasedHillClimbing$} \\ \hline
	
	&	&  $Avg$ & $Std$ & $Avg$ & $Algo$ &  $Avg$ & $Std$ & $Avg$ & $Algo$ \\
	
	&	& $(steps)$  & $(steps)$  & $(count)$ & $Failed(\%)$ & $(steps)$  &$(steps)$  & $(count)$ & $Failed(\%)$ \\	\hline
 
Astro-Ph	&	14845	&	8.66	&	9.82	&	0.37	&	1.21	&	5.70	&	1.40	&	0.00&	3.27	\\ \hline																			
Buzznet	&	101163	&	3.20	&	0.69	&	0.00	&	0.00	&	3.20	&	0.69	&	0.00	&	0.00	\\ \hline

Cond-Mat	&	13861	&	15.55	&	15.46	&	2.27	&	1.39	&	9.95	&	4.88	&	0.18	&	11.58	\\ \hline

DBLP	&	317080	&	213.13	&	84.89	&	9.79	&	5.45	&	118.76	&	5.34	&	0.00	&	79.60 \\ \hline
																			
Digg	&	261489	&	4.09	&	1.50	&	0.02	&	0.03	&	4.03	&	0.85	&	0.00	&	0.15	\\ \hline
																			
Enron	&	84384	&	5.08	&	0.97	&	0.00	&	0.00	&	5.07	&	0.93	&	0.00	&	0.04	\\ \hline																							
Facebook	&	63392	&	6.06	&	1.90	&	0.30	&	0.00	&	5.99	&	1.80	&	0.31	&	0.02	\\ \hline																					
FB-Wall	&	43953	&	9.04	&	4.48	&	0.59	&	0.00	&	8.65	&	4.12	&	0.48	&	0.04	\\ \hline
																	
Foursquare	&	639014	&	10.84	&	0.38	&	0.00	&	0.00	&	10.84	&	0.38&	0.00	&	0.00	\\ \hline
																			
Gowalla	&	196591	&	4.33	&	4.17	&	0.15	&	0.02	&	3.92	&	1.32	&	0.03	&	0.22	\\ \hline
																	
\end{tabular}	
}																	
\end{table*}

The proposed methods can be further improved by only computing the $h^2-index$ of the neighbors that can be considered for the next step of the crawler instead of computing the $h^2-index$ of all of its neighbors. A node having the degree lower than the $h^2-index$ of the current node cannot have the $h^2-index$ higher than the current node, so, all such neighbors can be discarded. This further fastens up the proposed method.


We also observe that in all considered real-world networks, the induced subgraph of all the top-ranked nodes (the highest $h^2-index$ nodes) is connected. So, once we hit one top node, all top nodes can be identified. All these nodes can be used to spread the information faster in the network.

\section{Section 3: Rank Estimation of a Node}\label{kshellsec3}

In this section, we employ structural behavior of $h^2-index$ to propose a fast rank estimation method. We observe that the curve of percentile rank versus $h^2-index$ follows a unique pattern in all the considered real-world networks as shown in Figure~\ref{ch4figrank}. The percentile rank of a node is computed as $PercentileRank(u)=\frac{n-R(u)+1}{n}*100$, where $n$ is the network size and $R(u)$ is the rank of node $u$ based on its $h^2-index$ value in the given network. The study of this curve shows the 4-parameter logistic equation to be a best fit. The equation of the curve is given as,

\begin{center}
\begin{equation}\label{kshelleq}
PercentileRank(u) = a_2 + \frac{a_1-a_2}{1+\left( \frac{h^2-index(u)}{x_0}\right) ^p},
\end{equation}
\end{center}
where $a_1$, $a_2$, $x_0$, and $p$ are parameters, and $p$ denotes slope of the logistic curve (also called hill's slope).

The plots are shown in Figure~\ref{ch4figrank}, where black colored circles depict the actual percentile rank of the nodes, and the red colored triangles show the best-fit curve using the $4$-parameter logistic equation. The best fit curve is plotted using scaled Levenberg-Marquardt algorithm \cite{more1978levenberg} with 1000 iterations and 0.0001 tolerance. It can be concluded from the plots that the logistic equation can be used to estimate the percentile rank of a node. Once we estimate the parameters of the equation, the percentile rank of a node can be computed in $O(1)$ time without computing the index values of all the nodes. This approach is similar to the heuristic method proposed for estimating the closeness rank of a node.

\vspace*{.5cm}

\textbf{Estimating Parameters of the Logistic Curve:} We estimate the parameters of the logistic curve by analyzing the curves for different networks. The estimated value of a parameter $x$ is represented by $x'$. In the logistic curve, $a_2$ is close to the highest rank value, so, it can be estimated as, $a_2'=100$. The value of $a_1$ is close to the minimum percentile rank. Using this observation, we estimate $a_1'=1$ and plots show that the estimated curve is close to the best fit. The value of $p$ is computed by averaging the slopes of the considered datasets using the estimated value of $a_1$ and $a_2$. Table~\ref{kshellpvalue} shows $p$ values using the estimated parameters ($a_1=1$ and $a_2=100$) and their average is $1.44$, so $p'=1.44$. $x_0$ is computed using scaled Levenberg-Marquardt algorithm \cite{more1978levenberg} for implementation, and its better estimator can be further proposed. 

\begin{table}[h]
\centering
\caption{Networks versus their p values using the estimated values of $a_1$ and $a_2$}
\label{kshellpvalue}
\begin{tabular}{|l|c|l|c|}
\hline
Network  & p value & Network  & p value \\ \hline
Astro-Ph & 1.49   & Enron  & 0.90   \\ \hline
Buzznet     & 1.26   & Facebook    & 1.37   \\ \hline
Cond-Mat     & 2.44   & FB-Wall & 1.56   \\ \hline
DBLP    & 2.08   & Foursquare  & 0.97   \\ \hline
Digg  & 0.91   & Gowalla & 1.40   \\ \hline
         &         & \textbf{Average}  & 1.44   \\ \hline
\end{tabular}
\end{table}

\subsection*{Discussion}

The percentile rank of a node is estimated using equation \ref{kshelleq}, so,\\ $Estimated \; PercentileRank(u) = a_2' + \frac{a_1'-a_2'}{1+\left( \frac{h^2-index(u)}{x_0'}\right) ^{p'}}$. The plots for the actual rank, best-fit rank, and estimated rank are shown in Figure~\ref{ch4figrank}. To measure the accuracy of the proposed method, we compute absolute error ($|Actual \; Percentile \; Rank - Computed \; Percentile \; Rank|$) for each $h^2-index$ value and find their average. The average error and standard deviation for the estimated rank using the best-fit and estimated parameters are shown in Table~\ref{kshellsec3err}. We further compute Kendall's Tau $(\tau)$, Pearson $(r)$, and Spearman $(\rho)$ correlation coefficients for the actual versus estimated ranks using the best-fit and estimated parameters. The results are shown in Table~\ref{kshellsec3corr} with the values rounded off by two decimal places. The results show that the logistic curve can be efficiently used to estimate the percentile rank. The percentile rank can be converted to actual rank if the network size or its estimated value is known. However, we have considered the percentile rank as it gives a complete idea about the relative rank of the node.

\begin{table*}[htp]
\centering
\caption{Absolute error for percentile ranking using best-fit and estimated curve}
\label{kshellsec3err}
\begin{tabular}{|l|c|c|c|c|}
\hline
Network &  \multicolumn{2}{|c|}{Best-Fit} & \multicolumn{2}{|c|}{Estimated} \\ \hline

	&	Avg. Error	&	Std. Dev &	Avg. Error	&	Std. Dev	\\	\hline
Astro-Ph	&	0.36	&	0.28	&	1.96	&	1.02	\\	\hline
Buzznet	&	5.73	&	6.13	&	1.73	&	1.91	\\	\hline
Cond-Mat	&	0.64	&	0.39	&	6.92	&	2.86	\\	\hline
DBLP	&	0.18	&	0.14	&	2.83	&	2.05	\\	\hline
Digg	&	0.36	&	0.29	&	0.95	&	1.31	\\	\hline
Enron	&	0.57	&	0.31	&	1.48	&	1.74	\\	\hline
Facebook	&	6.58	&	5.62	&	2.66	&	1.66	\\	\hline
Fb-Wall	&	0.65	&	0.43	&	2.43	&	0.88	\\	\hline
Foursquare	&	1.54	&	1.46	&	2.02	&	3.01	\\	\hline
Gowalla	&	0.19	&	0.12	&	0.64	&	0.51	\\	\hline
															
\end{tabular}	
\end{table*}

\begin{table*}[htp]
\centering
\caption{Correlation coefficients for percentile ranking using best-fit and estimated curve}
\label{kshellsec3corr}
\begin{tabular}{|l|c|c|c|c|c|c|}
\hline
Network &  \multicolumn{3}{|c|}{Best-Fit} & \multicolumn{3}{|c|}{Estimated} \\ \hline

	&	Kendall	&	Pearson	&	Spearman	&	Kendall	&	Pearson	&	Spearman	\\	\hline
Astro-Ph	&	1.00	&	1.00	&	1.00	&	1.00	&	1.00	&	1.00	\\	\hline
Buzznet	&	1.00	&	0.83	&	1.00	&	1.00	&	0.99	&	1.00	\\	\hline
Cond-Mat	&	1.00	&	1.00	&	1.00	&	1.00	&	0.99	&	1.00	\\	\hline
DBLP	&	1.00	&	1.00	&	1.00	&	1.00	&	0.99	&	1.00	\\	\hline
Digg	&	1.00	&	0.99	&	1.00	&	1.00	&	0.94	&	1.00	\\	\hline
Enron	&	1.00	&	0.99	&	1.00	&	1.00	&	0.93	&	1.00	\\	\hline
Facebook	&	1.00	&	0.88	&	1.00	&	1.00	&	0.99	&	1.00	\\	\hline
Fb-Wall	&	1.00	&	1.00	&	1.00	&	1.00	&	0.99	&	1.00	\\	\hline
Foursquare	&	1.00	&	0.96	&	1.00	&	1.00	&	0.92	&	1.00	\\	\hline
Gowalla	&	1.00	&	1.00	&	1.00	&	1.00	&	1.00	&	1.00	\\	\hline
																
\end{tabular}	
\end{table*}

\begin{figure*}[htp]
     \begin{center}
        \subfigure[Astro-Ph]{%
            \label{fig:first1}
            \includegraphics[width=0.45\textwidth]{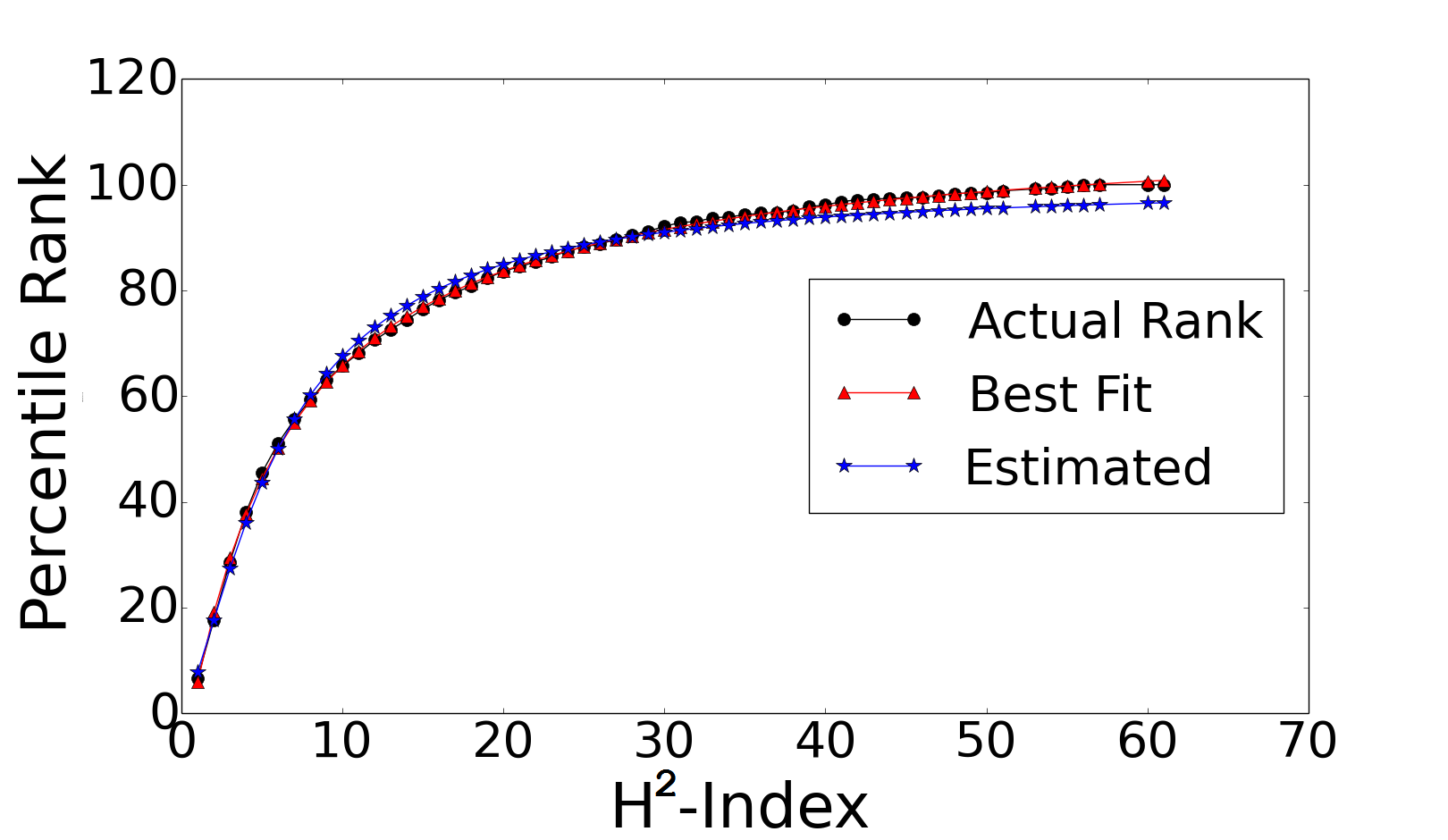}   
        }\quad%
        \subfigure[Buzznet]{%
           \label{fig:second2}
           \includegraphics[width=0.45\textwidth]{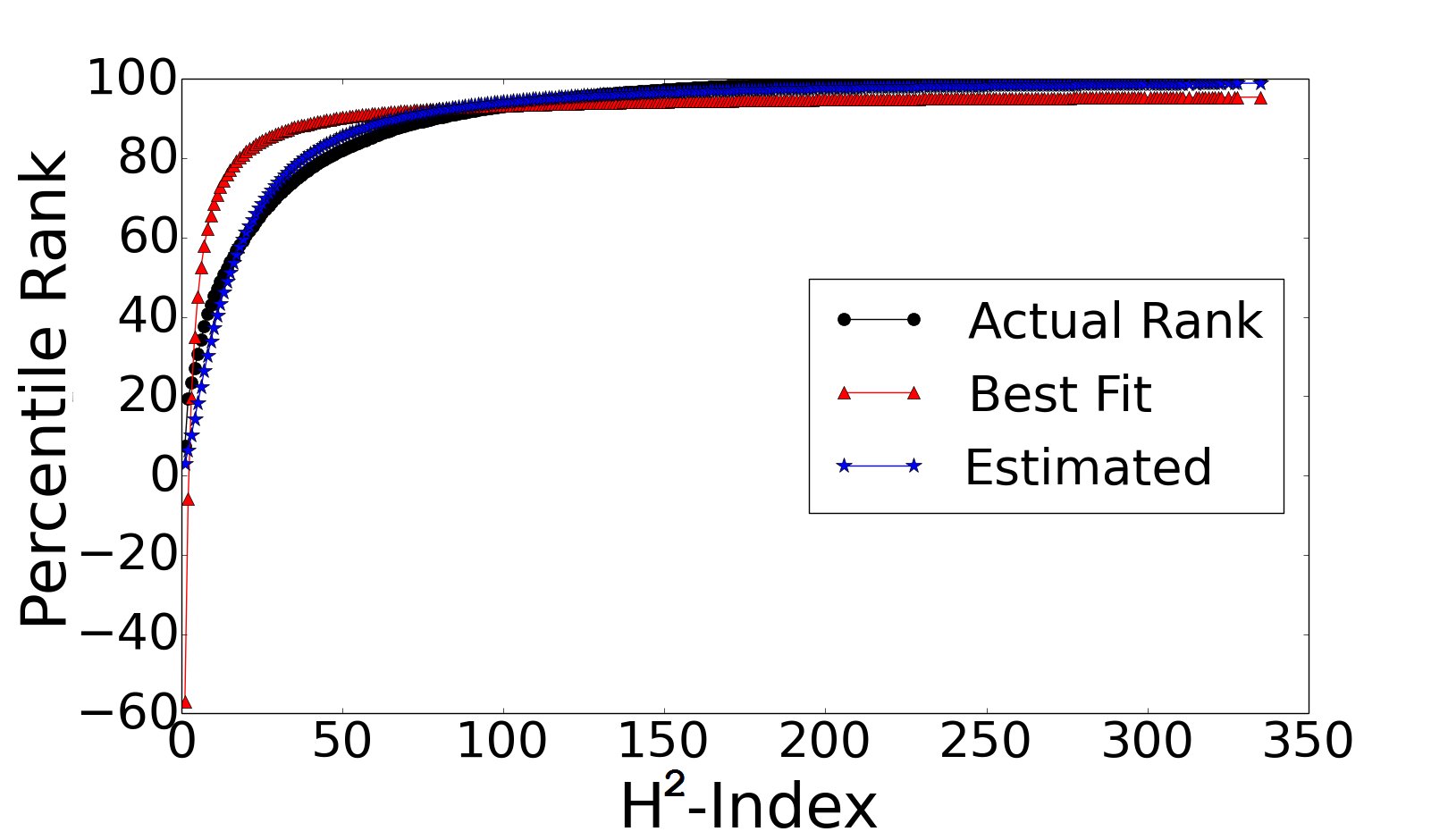}
        }\\
        \subfigure[Cond-Mat]{%
            \label{fig:first3}
            \includegraphics[width=0.45\textwidth]{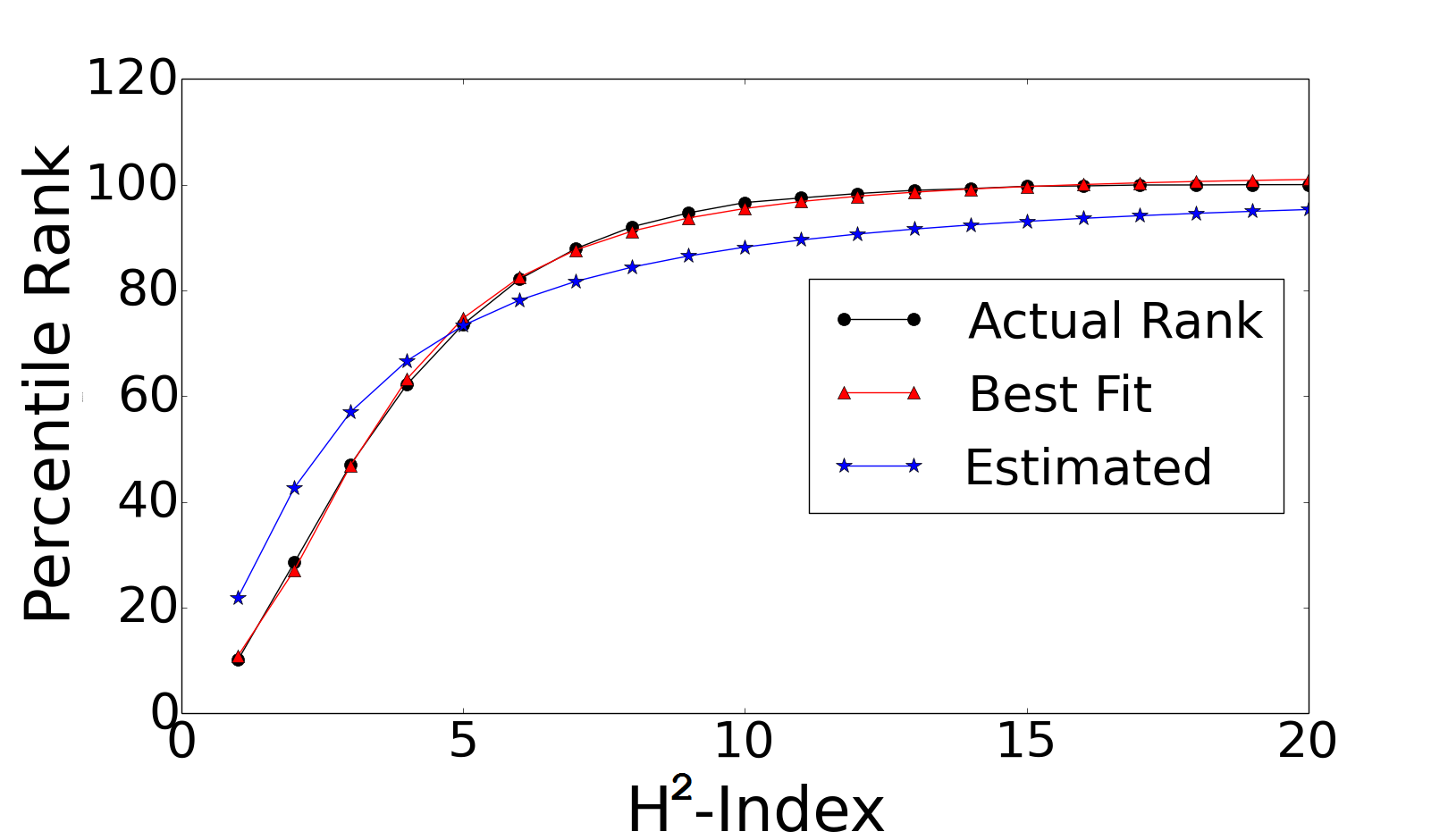}   
        }\quad%
        \subfigure[DBLP]{%
           \label{fig:second4}
           \includegraphics[width=0.45\textwidth]{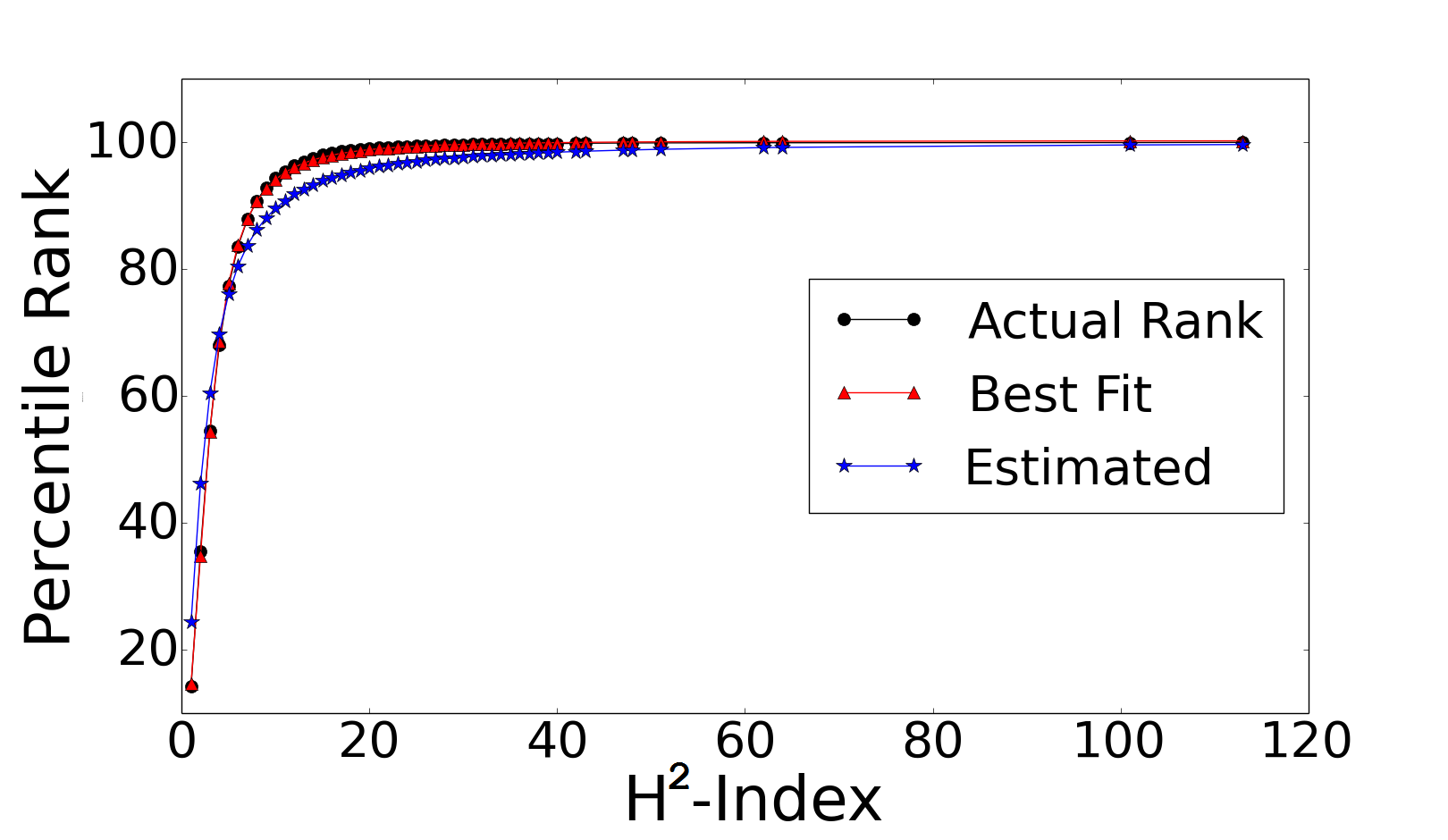}
        }\\
        \subfigure[Digg]{%
            \label{fig:first5}
            \includegraphics[width=0.45\textwidth]{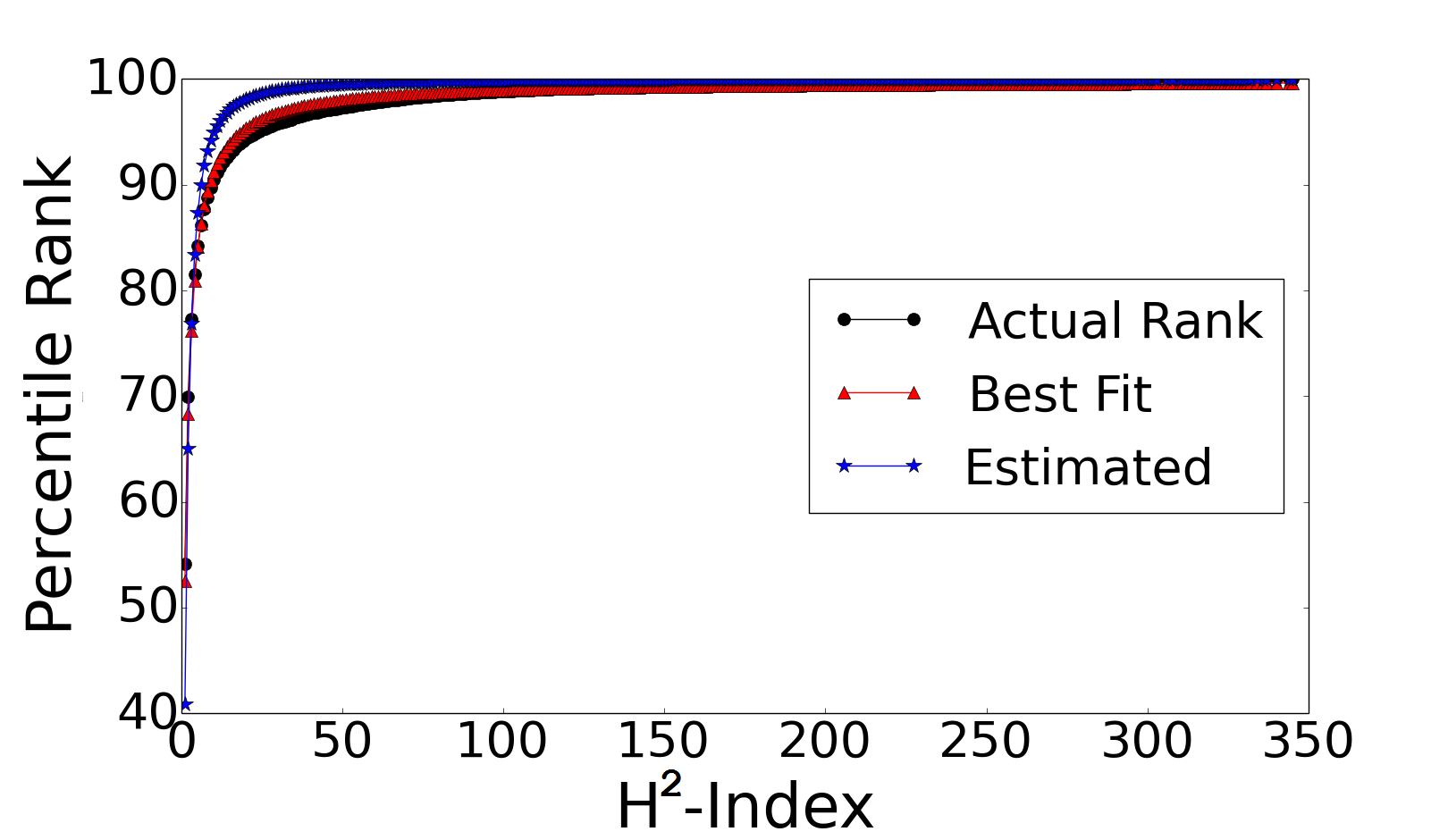}   
        }\quad%
        \subfigure[Enron]{%
           \label{fig:second6}
           \includegraphics[width=0.45\textwidth]{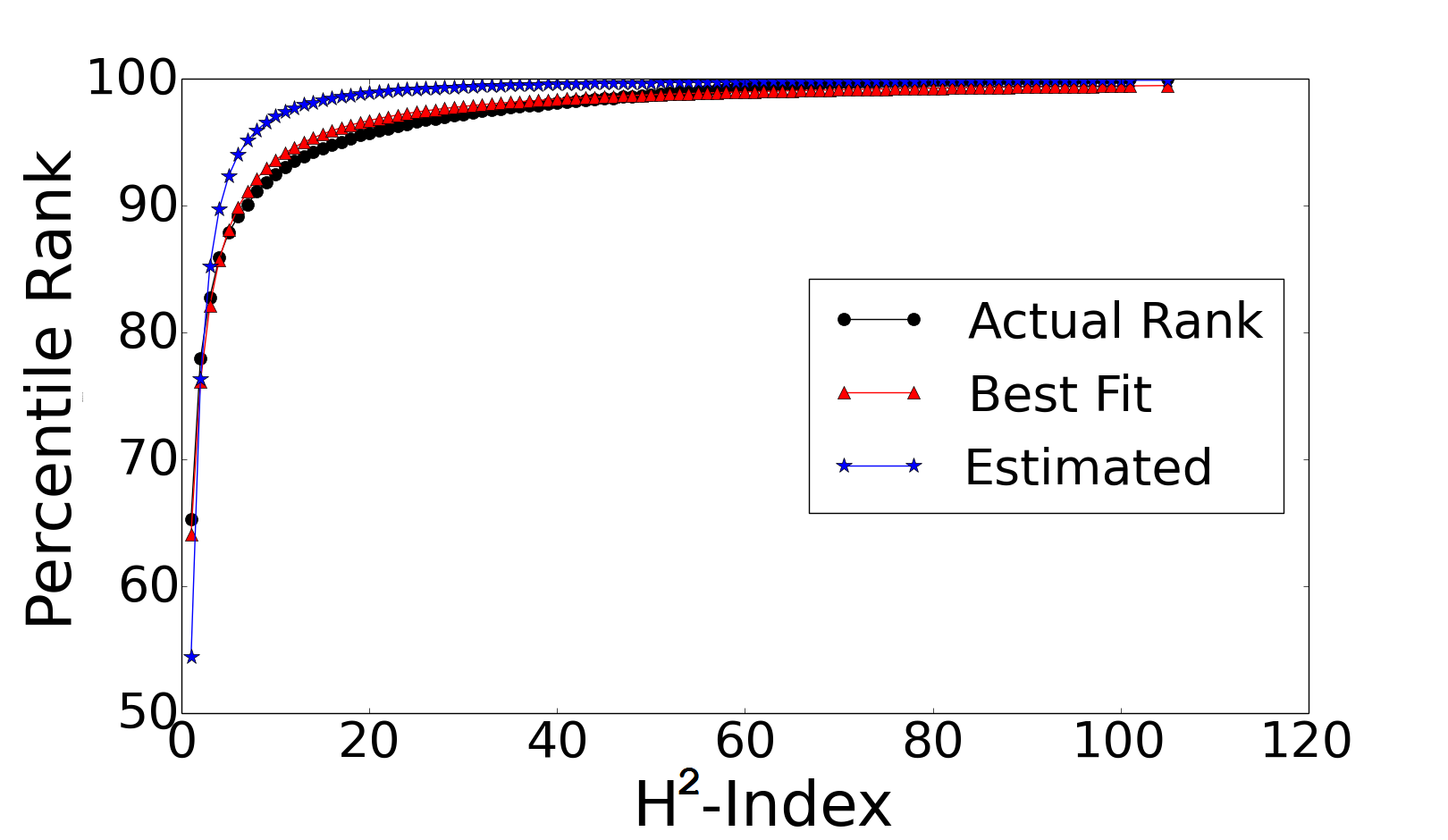}
        }\\
        \subfigure[Facebook]{%
            \label{fig:first7}
            \includegraphics[width=0.45\textwidth]{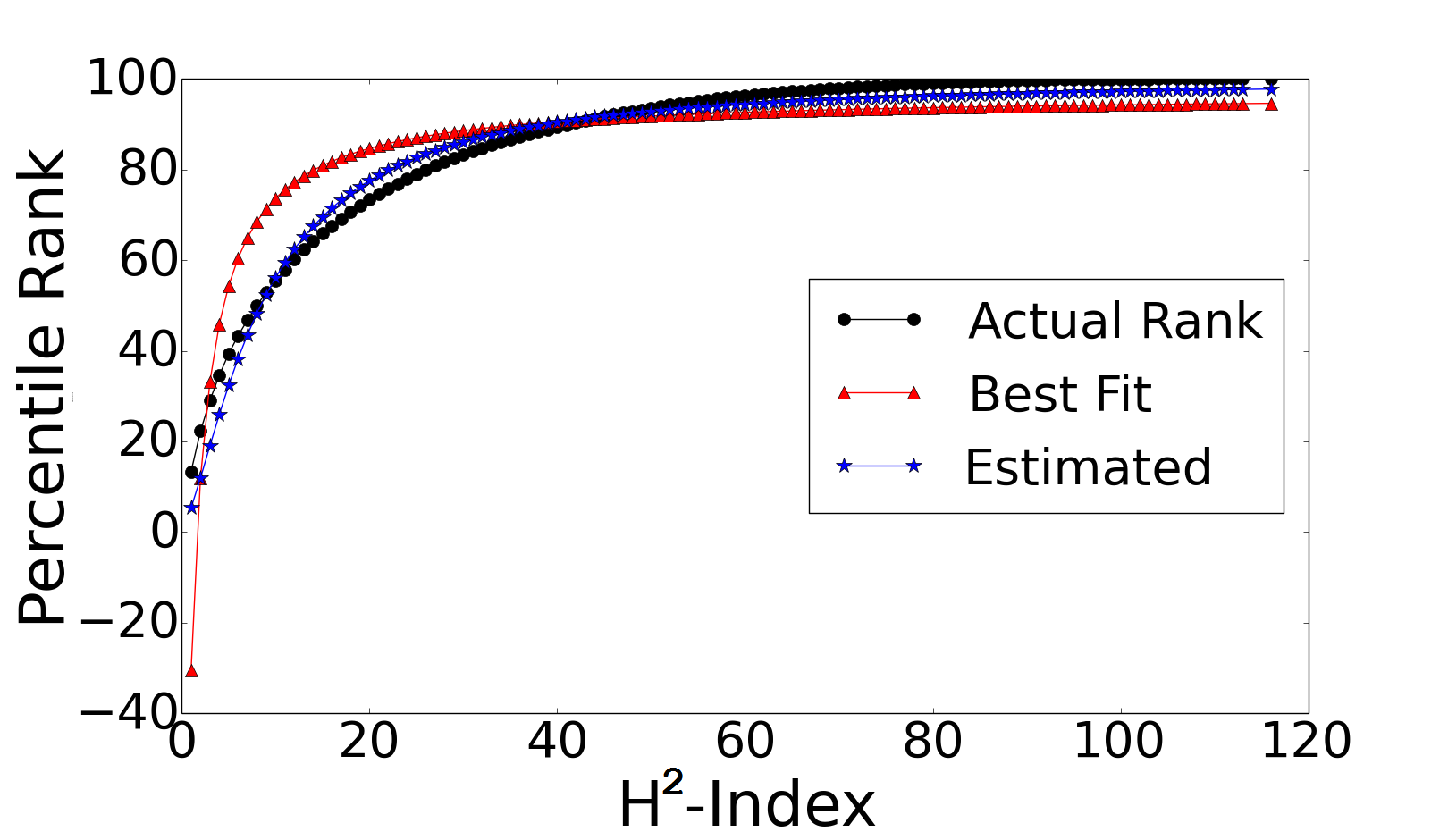}   
        }\quad%
        \subfigure[FB-Wall]{%
           \label{fig:second8}
           \includegraphics[width=0.45\textwidth]{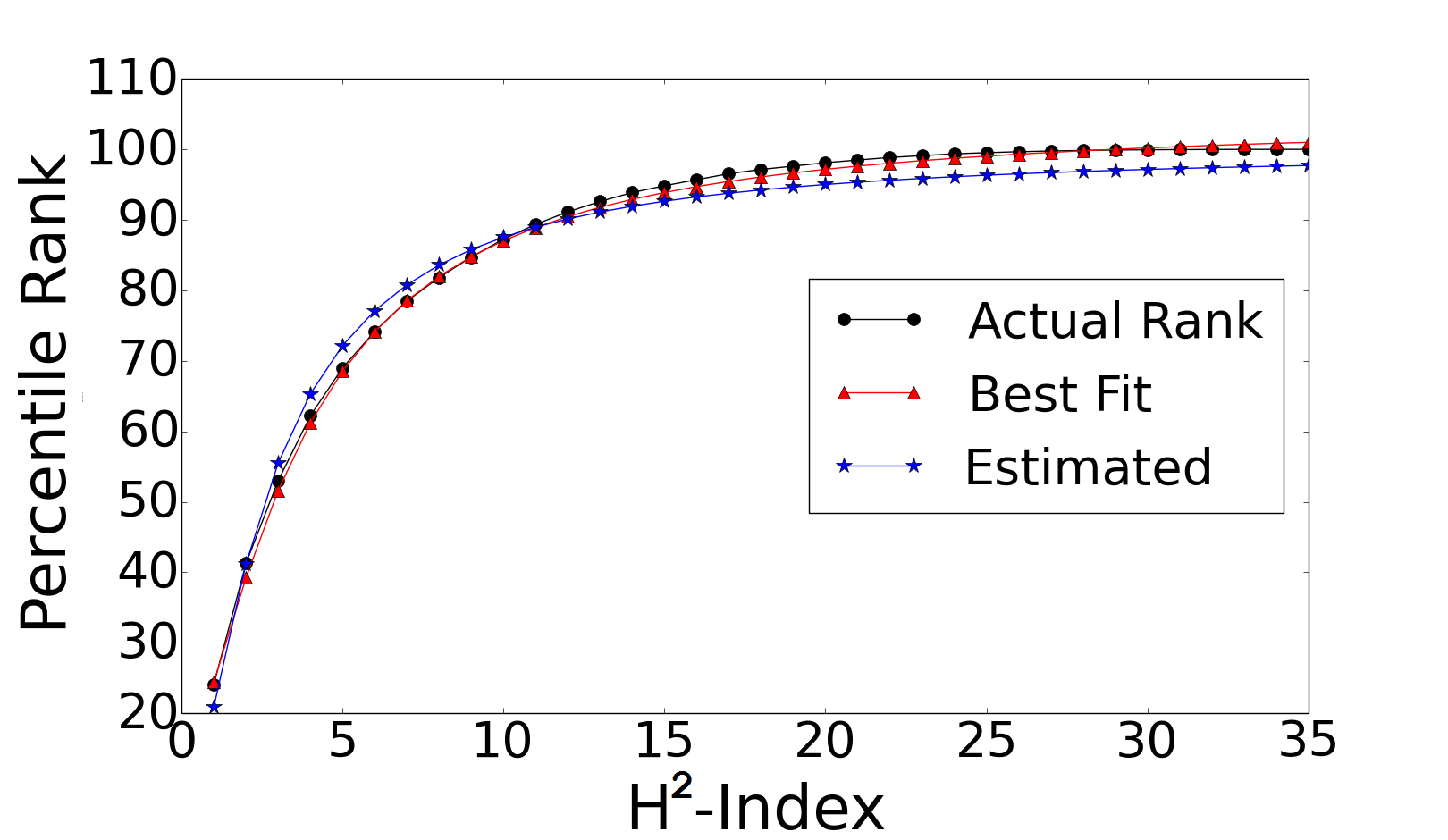}
        }\\
        \subfigure[Foursquare]{%
            \label{fig:first9}
            \includegraphics[width=0.45\textwidth]{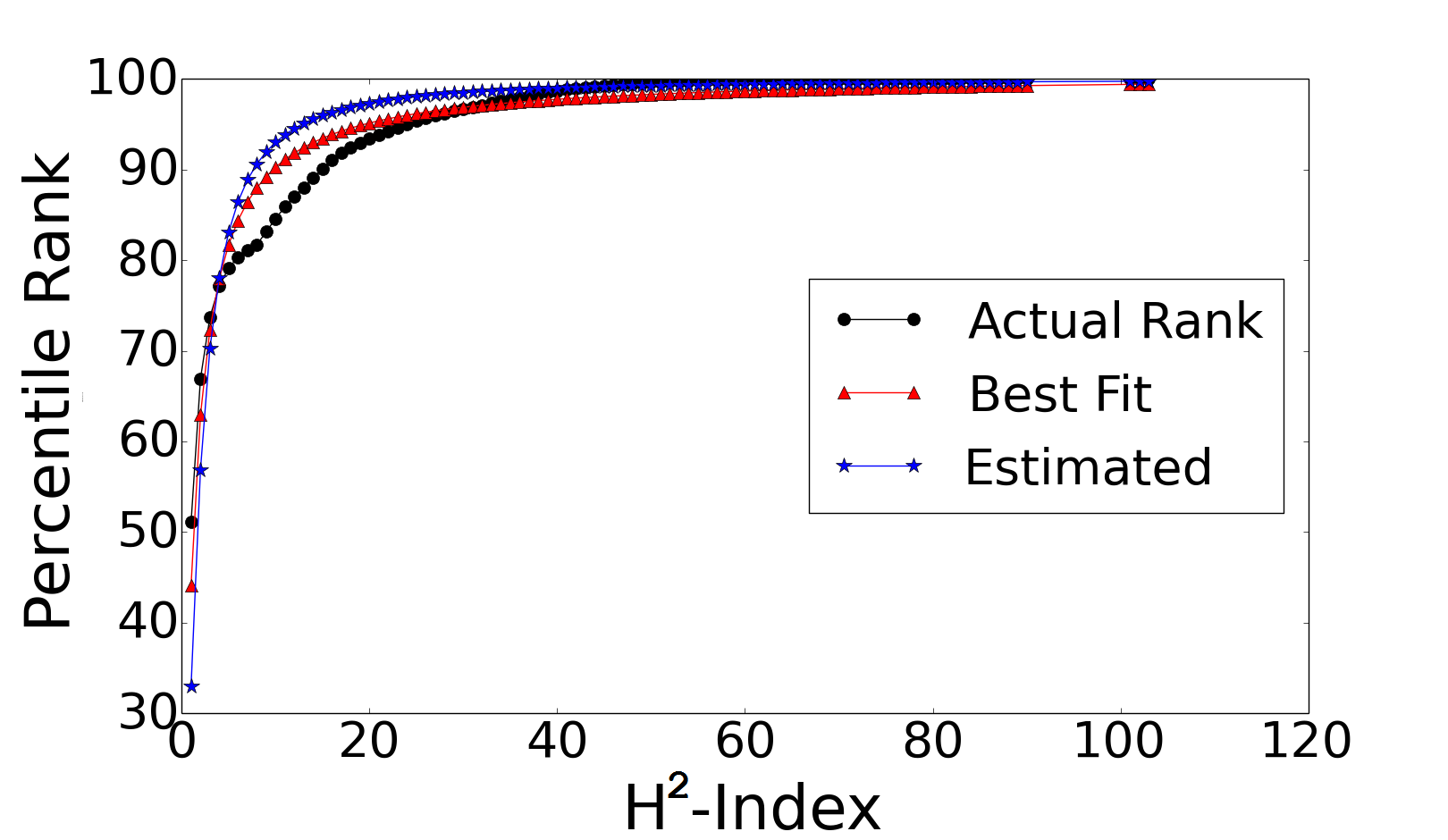}   
        }\quad%
        \subfigure[Gowalla]{%
           \label{fig:second9}
           \includegraphics[width=0.45\textwidth]{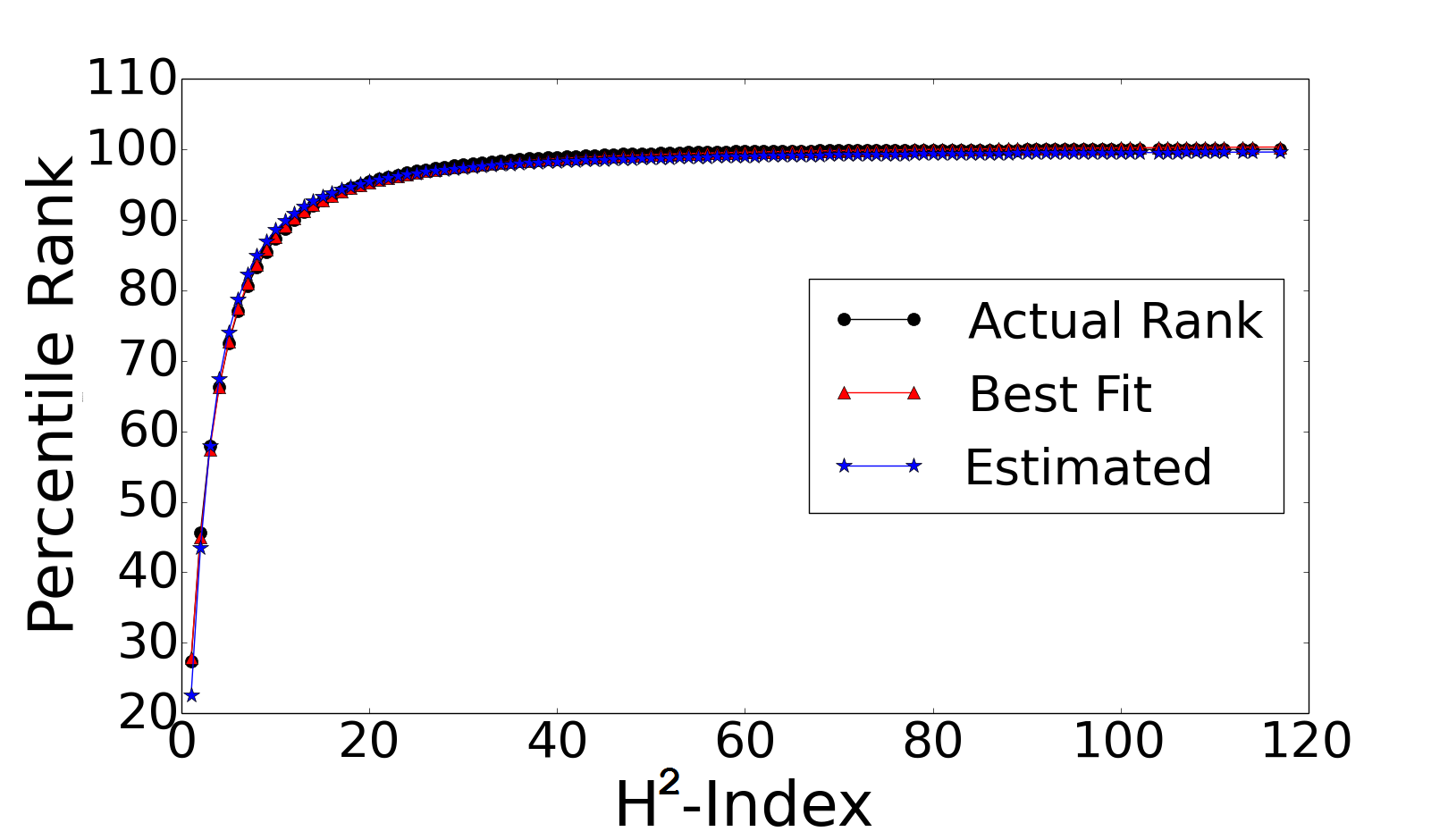}
        }\\
    \caption{Percentile Rank versus $h^2-index$}
   \label{ch4figrank}
   \end{center}
\end{figure*}

\section{Conclusion}\label{conclusion}

In this work, we have shown that $h^2-index$ is a good estimator of shell-index and can be used to identify influential nodes in real-world networks. We also observed that $h^2-index$ has better monotonicity than the shell-index and has better correlation with the spreading power of the nodes. We further proposed hill-climbing based methods to identify the top-ranked nodes. The results show that a top-ranked node can be found in a small number of steps. We also observed that the induced subgraph of all the top-ranked nodes is connected and once we find one top node, all the top nodes can be identified. A mathematical bound for estimating the number of steps to hit a top node in a given scale-free network can be further proposed. In the last section, we discussed a heuristic method to estimate the percentile rank of a node. The accuracy of the proposed method is computed using the absolute error and correlation coefficients. One can further propose better methods for estimating the parameters of the logistic curve which will improve the accuracy of the rank estimation.

\vspace*{.5cm}

\textbf{Acknowledgement}: Authors would like to thank IIT Ropar HPC committee for providing the resources for performing the experiments.

\bibliographystyle{unsrt}
\bibliography{mybib}

\end{document}